%% file: max-min-greedy-matching.tex
\documentclass[11pt,letterpaper]{article}

\usepackage{amsmath}
\usepackage{amsthm}
\usepackage{amssymb}
\usepackage{mathtools}
\usepackage{dsfont}
\usepackage[dvipsnames]{xcolor}
\usepackage{xfrac}

\newcommand{\comment}[1]{}
\newcommand{\Omit}[1]{}

\newtheorem{thm}{Theorem}
\newtheorem{lemma}{Lemma}
\numberwithin{lemma}{section}
\newtheorem{corollary}[lemma]{Corollary}

\newtheorem{proposition}[lemma]{Proposition}

\newtheorem{fact}[lemma]{Fact}

\newenvironment{AvoidOverfullParagraph}[0]
{\sloppy\ignorespaces}
{\par\fussy\ignorespacesafterend}

\begin{document}

\title{Max-Min Greedy Matching}

\author{
Alon Eden\thanks{Tel-Aviv University; {\tt alonarden@gmail.com}}
\and
Uriel Feige\thanks{Weizmann Institute of Science; {\tt uriel.feige@weizmann.ac.il}}
\and
Michal Feldman\thanks{Tel-Aviv University and Microsoft Research Israel; {\tt mfeldman@tau.ac.il}}
}
\date{}

\maketitle

\begin{abstract}
\input{abstract}
\end{abstract}

\input{intro}

\input{main-result}

\input{regular}

\input{random}

\section*{Acknowledgments}

A substantial part of this work was conducted in Microsoft Research, Herzeliya, Israel. The work of U. Feige was supported in part by the Israel Science Foundation
(grant No. 1388/16). The work of M. Feldman and A. Eden was partially supported by the European Research Council under the European Union's Seventh Framework Programme (FP7/2007-2013) / ERC grant agreement number 337122, and by the Israel Science Foundation (grant number 317/17).
We are grateful to Amos Fiat and Sella Nevo for numerous discussions that contributed significantly to the ideas presented in this paper.
We also thank Robert Kleinberg for helpful discussions.

\bibliographystyle{plain}
\bibliography{permutation-matching}

\medskip

\appendix
\section*{APPENDIX}
\setcounter{section}{0}

\input{hamiltonian}

\input{iterative}

\input{app-additional-examples}

\end{document}

%% file: abstract.tex
A bipartite graph $G(U,V;E)$ that admits a perfect matching is given.
One player imposes a permutation $\pi$ over $V$, the other player imposes a permutation $\sigma$ over $U$.
In the greedy matching algorithm, vertices of $U$ arrive in order $\sigma$ and each vertex is matched to the lowest (under $\pi$) yet unmatched neighbor in $V$ (or left unmatched, if all its neighbors are already matched).
The obtained matching is maximal, thus matches at least a half of the vertices.
The max-min greedy matching problem asks: suppose the first (max) player reveals $\pi$, and the second (min) player responds with the worst possible $\sigma$ for $\pi$, does there exist a permutation $\pi$ ensuring to match strictly more than a half of the vertices? Can such a permutation be computed in polynomial time?

The main result of this paper is an affirmative answer for this question: we show that there exists a polytime algorithm to compute $\pi$ for which for every $\sigma$ at least $\rho > 0.51$ fraction of the vertices of $V$ are matched.
We provide additional lower and upper bounds for special families of graphs, including regular and Hamiltonian.
Interestingly, even for regular graphs with arbitrarily large degree (implying a large number of disjoint perfect matchings), there is no $\pi$ ensuring to match more than a fraction $8/9$ of the vertices.

The max-min greedy matching problem solves an open problem regarding the welfare guarantees attainable by pricing in sequential markets with binary unit-demand valuations.
In addition, it has implications for the size of the unique stable matching in markets with global preferences, subject to the graph structure. 

%% file: intro.tex
\section{Introduction}
\label{sec:introduction}

In a typical Internet advertising scenario, advertisers arrive sequentially and are matched in an online fashion to a set of ad slots.
The domain of Internet advertising has led to a surge of interest in the study of online matching problems.
Indeed, a large number of new models of online matching problems have been introduced in the last decade and new algorithmic techniques have emerged (see the survey of Mehta \cite{Mehta13} and references therein). The new problems are both theoretically elegant and practically relevant.

In an online bipartite matching problem, there is a set $V$ of right-side vertices, known in advance, and a set $U$ of left-side vertices, which arrive in an online fashion.
Upon the arrival of a vertex $u \in U$, the weights on the edges from $u$ to vertices in $V$ are revealed, and an immediate and irrevocable matching decision should be made (possibly leaving $u$ unmatched).
The goal is to maximize the total weight of the obtained matching.
In the Internet advertising analogue, $V$ is the set of items (ad slots) and $U$ is the set of buyers (advertisers).
Upon the arrival of a buyer, her values to all items are revealed and a matching decision should be made.
The total weight of the obtained matching is termed the welfare of the matching.
The performance of an online matching algorithm ALG is typically quantified by the {\em competitive ratio} of ALG, defined as the ratio of the expected welfare obtained by ALG (in the worst arrival order) and the maximum weighted matching in hindsight.

Feldman et al. \cite{FeldmanGL15} study the design of pricing mechanisms in {\em Bayesian} online matching problems; i.e., where the valuation of every buyer is drawn independently from a known distribution.
The designer assigns item prices, based on the known prior, and buyers arrive in an adversarial order (after observing the prices), each purchasing the item that maximizes her utility (defined as the difference between value and price).
It is shown that one can set item prices that guarantee a competitive ratio of $1/2$.
In fact, the last result holds in a much more general setting than matching, where buyers have submodular\footnote{A valuation is said to be submodular if for every two sets $S,T$, $v(S)+v(T) \geq v(S\cup T) + v(S \cap T)$,.} valuations over bundles of items, and each arriving buyer purchases a bundle that maximizes her utility (defined as the difference between the value of the bundle and the sum of item prices in the bundle).

The bound of $1/2$ has been shown to be tight in the Bayesian setting \cite{FeldmanGL15}.
I.e., if the designer knows the distributions from which values are drawn but not the realized values, then no item prices can obtain better than half the optimal welfare in the worst case.
A natural question to ask is whether this ratio can be improved in {\em full information} settings; i.e., in scenarios where the designer knows the realized values of the buyers from the outset.
The full information assumption is sensible in repeated markets or in markets where the stakes are high and the designer may invest in learning the demand in the market before setting prices.
Concretely, the problem is the following: in scenarios where the graph is fully known, does there exist item prices that guarantee strictly more than half the optimal welfare, for any arrival order $\sigma$?
This question has been open until now.
Not only has it been open for general combinatorial auctions with submodular valuations, it has been open even for unit-demand buyers, and even if all individual values are in $\{0,1\}$ (henceforth referred to as $\{0,1\}$ unit demand valuations).
This scenario has the following description as an online bipartite matching problem.

\vspace{0.1in}

\noindent {\bf Max-min greedy matching.}
Let $G(U,V;E)$ be an (unweighted) bipartite graph with $n$ vertices in each side.
Assume w.l.o.g. that $G$ admits a perfect matching.
For every vertex $w \in U \cup V$, let $N(w)$ denote the set of neighbors of $w$.
Consider a permutation $\sigma$ over $U$ and a permutation $\pi$ over $V$.
In the greedy matching algorithm, vertices of $U$ arrive in order $\sigma$ and each vertex $u \in U$ is matched to the lowest (under $\pi$) yet unmatched $v \in N(u)$ (or left unmatched, if all $N(u)$ is already matched).
Let $M_G[\sigma,\pi]$ denote the greedy matching obtained by using permutations $\sigma$ and $\pi$.
Let $\rho[G] = \frac{1}{n} \max_{\pi} \min_{\sigma}[|M_G[\sigma,\pi]|]$, and let $\rho = \min_G[\rho[G]]$. 
The question is whether $\rho > 1/2$: is there some $\rho > 1/2$ such that every bipartite graph (that admits a perfect matching) admits a permutation $\pi$ satisfying $\min_{\sigma}[|M_G[\sigma,\pi]|] \ge \rho n$?

Every greedy matching is a maximal matching. Hence for {\em every} $\pi$ the obtained matching is of size at least $n/2$ (and this is true even if the graph is unknown).
Moreover, for families of graphs where every maximal matching is greater than $n/2$ (e.g., for random $d$-regular graphs, almost surely every maximal matching has size at least $n - O(\frac{n \log d}{d})$), the answer to the above question is affirmative.

\vspace{0.1in}

Besides its implication to pricing applications, we believe that the max-min greedy matching problem is interesting in its own right as a variant of an online bipartite matching.
Moreover, this problem has an interpretation also as a static matching problem, in the context of stable matching.
In what follows, we provide three different interpretations to the max-min greedy matching problem. In all three interpretations the question is whether it is possible to obtain a matching that is strictly greater than a half of the vertices.

\vspace{0.1in}
\noindent {\bf Interpretation 1: pricing / item ordering.} One can verify that in scenarios with $\{0,1\}$ unit demand valuations, the max-min greedy matching problem is exactly equivalent to the online pricing problem described above. Here, $U$ is the set of buyers, $V$ the set of items, and the designer determines item prices that induce an ordering $\pi$ over the items, followed by an arbitrary arrival order $\sigma$ of the buyers. The competitive ratio with respect to the online pricing problem equals exactly the value of $\rho$ associated with the max-min greedy matching problem.

An equivalent scenario is one where one player controls the items, the other player controls the buyers.
In each step, the items player offers an item, and the buyers player, upon seeing the item, allocates the item to one of the buyers that wants the item (if there is any), and that buyer leaves.
The items player is non-adaptive (plays blindfolded, without seeing which buyers remain\footnote{We note that when the items player is adaptive (chooses the next item based on what happened in the past), the items player can ensure a perfect matching. This is done as follows: in each step, find a minimal tight set of items, and offer an arbitrary item from that set. Here, a set of items is tight if the number of buyers that want items in the set is equal to the size of the set.}
).
The size of the matching that can be guaranteed by the items player is equivalent to the max-min greedy matching problem.

\vspace{0.1in}
\noindent {\bf Interpretation 2: buyer ordering.}
An equivalent formulation of the problem is one where the permutation $\pi$ is imposed over the buyers rather than over the items.
The buyers then arrive in the order of $\pi$, each taking an arbitrary item she wants.
One can verify that the size of the matching that can be guaranteed by an ordering over the buyers is equivalent to the max-min greedy matching problem.

\vspace{0.1in}
\noindent {\bf Interpretation 3: stable matching.}
A third motivation comes from the world of stable matching \cite{GS62}.
This interpretation is static and has no online component.
In a stable matching scenario, every vertex in $U$ has a preference order over (possibly a subset of) the vertices in $V$, and every vertex in $V$ has a preference order over (possibly a subset of) the vertices in $U$.
Given a matching $M$, a pair of vertices $u \in U, v \in V$ is said to constitute a blocking pair if $u$ and $v$ are not matched in $M$, but they both prefer each other over their partners in $M$.
A matching is said to be stable if no pair of vertices $u,v$ constitutes a blocking pair.
Given a graph $G$ and permutations $\pi$ and $\sigma$ over $V$ and $U$ respectively, a stable matching problem is induced, where $\pi$ and $\sigma$ correspond to global preference orders over $V$ and $U$, respectively, subject to the graph structure.
That is, every vertex in $U$ (resp., $V$) prefers a neighbor of lower rank according to $\pi$ (resp., according to $\sigma$).
One can easily verify that global preference orders imply the existence of a unique stable matching.
Moreover, the outcome of the greedy matching process is the unique stable matching.
Thus, our question is equivalent to asking whether there exists a permutation $\pi$ over $V$ such that for every permutation $\sigma$ over $U$ the unique stable matching with respect to $\pi$ and $\sigma$ obtains a large matching.

\vspace{0.1in}

The reader familiar with the seminal work of Karp et al. \cite{KarpVV1990} might wonder whether their tight $1-1/e$ competitive ratio applies to our problem. After all, the algorithm in \cite{KarpVV1990} is a ranking algorithm, followed by the greedy matching process.
The answer is: neither the upper nor the lower bound applies.
The lower bound does not apply since the algorithm in \cite{KarpVV1990} imposes a {\em random} permutation $\pi$, and the adversary sets the permutation $\sigma$ without observing the realization of $\pi$.
This is in contrast to our model, where $\sigma$ is set after observing the realization of $\pi$.
The upper bound does not apply either since in \cite{KarpVV1990} every vertex $u \in U$ arrives with its incident edges, which are not known in advance.
In contrast, in our model, the graph is known from the outset. It is only the order of arrival that is unknown.

\begin{figure}[h!]
\begin{center}
	\includegraphics[scale=.4]{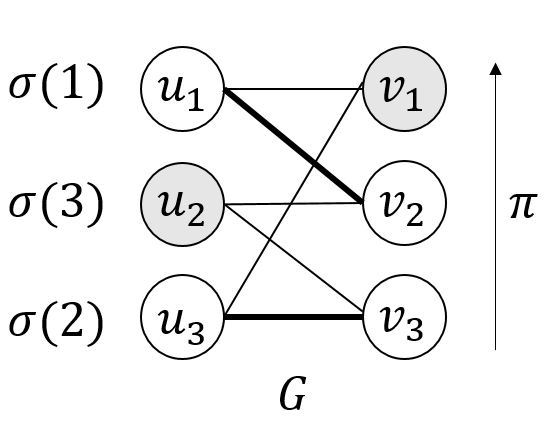}
	\caption{For every permutation $\pi$ there exists a permutation $\sigma$ that matches only 2 of the 3 vertices. Thick edges are in the matching; gray vertices are unmatched.}
\end{center}
\label{fig:two-thirds}
\end{figure}

Cohen Addad et al. \cite{Cohen-AddadEFF16} show an example of a bipartite graph $G=(V,U;E)$ with 3 vertices on each side that admits a perfect matching for which no permutation $\pi$ can guarantee to match more than 2 vertices in the worst case; that is $\rho[G] \leq 2/3$.
In their example (depicted in Figure \ref{fig:two-thirds}), $U = \{u_1,u_2,u_3\}$, $V = \{v_1,v_2,v_3\}$, $u_1$ is connected to $v_1,v_2$, $u_2$ is connected to $v_2,v_3$ and $u_3$ is connected to $v_3,v_1$. Without loss of generality, suppose $\pi=(v_3,v_2,v_1)$. For $\sigma=(u_1,u_3,u_2)$, $u_1$ is matched to $v_2$, $u_3$ is matched to $v_3$, and $u_2$ is left unmatched, resulting in a matching of size $2$.
This example shows that in general $1/2 \leq \rho \leq 2/3$.
However, it may still be the case that $\rho > 2/3$ for special classes of graphs.

\subsection{Our Results and Open Problems}

Our main result resolves the open problem in the affirmative:

\vspace{0.1in}
\noindent {\bf Theorem 1 [main theorem]:}
It holds that $\rho \geq \frac{1}{2} + \frac{1}{86} > 0.51$.
\vspace{0.1in}

The significance of this result is that $1/2$ is not the optimal answer.
We believe that further improvements are possible. In fact, for Hamiltonian graphs the proof of Theorem~1 implies that  $\rho \ge \frac{5}{9} - \frac{1}{n}$.

The proof method is quite involved; it is natural to ask whether simpler approaches may work.
A first attempt would be to check whether a random permutation $\pi$ obtains the desired result (in expectation).
The performance of a random permutation is interesting for an additional reason: it is the performance in scenarios where the graph structure is unknown to the designer.
Unfortunately, there exists a bipartite graph $G$, even one where all vertices have high degree, for which a random permutation matches no more than a fraction $1/2 + o(1)$ of the vertices (see Section ~\ref{sec:random}).

In contrast, we show that in the case of hamiltonian graphs a random permutation guarantees a competitive ratio strictly greater than $1/2$ (Section~\ref{sec:random}).
A similar proof approach applies to regular graphs as well.

\vspace{0.1in}
\noindent {\bf Theorem 2 [random permutation]:}
There is some $\rho_0 > \frac{1}{2}$ such that for every Hamiltonian graph $G$, regardless of $n$,
a random permutation $\pi$ results in $\rho \geq \rho_0$.
Similarly, there is some constant $\rho_0 > \frac{1}{2}$ such that for every $d$-regular graph $G$, regardless of $d,n$, a random permutation $\pi$ results in $\rho \geq \rho_0$.
\vspace{0.1in}

A second attempt would be to iteratively ``upgrade" unmatched vertices, with the hope that the iterative process will reach a state where many vertices will be matched.
That is, in every iteration consider the worst order $\sigma$ for the current permutation $\pi$ and move all unmatched vertices (in the matching induced by $(\pi,\sigma)$) to be ranked lower in $\pi$.
In Appendix~\ref{sec:iterative} we show that this process can go on for $\log n$ iterations before reaching a permutation that matches more than a half of the vertices.
This fact gives some indication that establishing a proof using this operator might be difficult.

Finally, we show that even for regular graphs with an arbitrarily large degree (in particular, graphs with a large number of disjoint perfect matchings), one cannot hope to get a better fraction than $8/9$.

\vspace{0.1in}
\noindent {\bf Theorem 3 [regular graphs]:}
For $d$-regulars bipartite graphs, $\rho \ge \frac{5}{9} - O(\frac{1}{\sqrt{d}})$.
On the other hand, for every integer $d \ge 1$, there is a regular graph $G_{d}$ of even degree $2d$ such that $\rho(G_{d}) \le \frac{8}{9}$.
\vspace{0.1in}


\vspace{0.1in}
\noindent {\bf Main open problem.}
Our analysis and results leave some interesting open problems for future research.
We show that $\rho$ is greater than half; finding a tight bound remains open. This is mostly interesting for general graphs, but also for special classes of graphs, such as regular and Hamiltonian ones.

\subsection{Related Work}


Since the seminal paper of Karp et al. \cite{KarpVV1990} there has been a long line of research on different variants of online bipartite matching.
The problem has gained a lot of interest in the economics and computation community due to its applicability to matching problems and Internet advertising problems (see the survey in \cite{Mehta13} and references therein).
The analysis of online bipartite matching is divided along various axes: the graph structure (e.g., regular vs. non-regular graphs), whether the vertices arrive in an adversarial or random order, and whether the algorithm is deterministic or random.
For adversarial arrival order, the deterministic greedy algorithm gives a competitive ratio of at least $1/2$, and this is tight for deterministic algorithms. For random algorithms, Karp et al. \cite{KarpVV1990} give a tight $1-1/e$ bound using random ranking.
For $d$-regular graphs, Cohen and Wajc \cite{CohenW18} have recently presented a random algorithm that obtains $1-O(\sqrt{\log d}/\sqrt{d})$ in expectation, and a lower bound of $1-O(1/\sqrt{d})$.
Under random arrival order, the deterministic greedy algorithm gives $1-1/e$, and no deterministic algorithm can obtain more than $3/4$ \cite{Goel2008}. The lower bound is obtained by a reduction from the model in \cite{KarpVV1990}.
Random ranking obtains at least $0.696$ of the optimal matching \cite{Mahdian2011} and at most $0.727$ \cite{Karande2011}.
More generally, no random algorithm can obtain more than $0.823$ \cite{Manshadi2012}.
In all of these results, the graph is unknown upfront.

Our work is also related to the recent body of literature on simple, approximately optimal truthful mechanisms \cite{Hartline2009}.
Motivated by the fact that in real-life situations one is often willing to trade optimality for simplicity, the study of simple mechanisms has gained a lot of interest in the literature on algorithmic mechanism design.
One of the simplest forms of mechanisms is that of posted price mechanisms, where prices are associated with items and agents buy their most preferred bundles as they arrive.
Pricing mechanisms have many advantages: they are simple, straightforward, allow for asynchronous arrival and departure of buyers, and are truthful in a strong sense (termed {\em obviously} strategyproof \cite{Li2017}).
Various forms of posted price mechanisms for welfare maximization have been proposed for various combinatorial settings \cite{FeldmanGL15,DuttingFKL16,Lucier17,EzraFRS17}. These mechanisms are divided along several axes, such as item vs. bundle pricing, static vs. dynamic pricing, and anonymous vs. personalized pricing.
For any market with submodular valuations, one can obtain $1/2$ of the optimal welfare by static, anonymous item prices \cite{FeldmanGL15}.
As mentioned in the introduction, until the current paper, no better results than $1/2$ were known even for markets with unit-demand valuations with $\{0,1\}$ individual values.
For a market with $m$ identical items, there exists a pricing scheme that obtains at least $5/7-1/m$ of the optimal welfare for submodular valuations \cite{EzraFRS17}.

%% file: main-result.tex
\section{Proof of Main Result}

Given a graph $G(U,V;E)$ with a perfect matching $M$ (in which $u_i \in U$ is matched with $v_i \in V$ for every $1 \le i \le n$), we consider the following directed graph $H(W,D)$.  $W$ contains $n$ vertices, where vertex $w_i$ corresponds to the edge $(u_i,v_i)$ in $G$. Directed edge (arc) $(w_i,w_j)$ is in $D$ if $(u_i,v_j) \in E$. We refer to $H(W,D)$ as the {\em spoiling graph} for $G$, because arc $(w_i,w_j) \in D$ allows for the possibility that edge $(u_i,v_j)$ is chosen into a matching $M'$ in $G$, spoiling for $v_i$ the possibility of being matched to $u_i$ (which is the matching edge used by $M$)\footnote{A similar graph is considered also in \cite{Cohen-AddadEFF16} and \cite{HsuMRRV16}.}. Given an arc $(w_i,w_j)$, we may refer to $v_j$ as a {\em spoiling vertex} for $v_i$.

A {\em path cover} of $H$ is a collection of vertex disjoint directed paths that covers all vertices in $W$. The length (number of vertices) of a path $P$ is denoted by $|P|$. We allow a path cover to contain isolated vertices (paths of length~1), and hence $H$ necessarily has a path cover. Given a path cover, we consider the following operations:

\begin{enumerate}

\item {\em Path merging}: if $H$ has an arc from the end of one path to the start of another path, the two paths can be merged into one longer path.

\item {\em Path unbalancing}: consider any two paths $P_1$ and $P_2$ with $|P_1| \ge |P_2|$. If $H$ contains an arc from the the first vertex (say $w$) of $P_2$ to the first vertex of $P_1$, we may remove $w$ from $P_2$ and append it at the beginning of $P_1$. Likewise, if $H$ contains an arc from the last vertex of $P_1$ to the last vertex (say $w$) of $P_2$, we may remove $w$ from $P_2$ and append it at the end of $P_1$.

\item {\em Rotation}: if there is a path $P$ (say, $w_1, \ldots, w_{\ell}$) such that in $H$ there is an arc $(w_{\ell},w_1)$ from the end of the path to its beginning, then we may add this arc to $P$ (obtaining a cycle), and then remove any single arc from the resulting cycle to get a path $P'$. Observe that $P'$ and $P$ have the same vertex set, but they differ in their starting vertex along the cycle $w_1, \ldots, w_{\ell}, w_1$.

\end{enumerate}

A path cover is {\em maximal} if no path merging operation and no path unbalancing can be applied to it, not even after performing rotation operations. Observe that if a path cover is not maximal, then at most two rotation operations are required in order to be able to perform either a path merging or a path unbalancing operation. Consequently, one can check in polynomial time whether a path cover is maximal.

Every directed graph has path covers that are maximal, because the path merging and path unbalancing operations increase the sum of squares of the lengths of paths. Moreover, a maximal path cover can be found in polynomial time, starting with the trivial path cover in which all vertices of $W$ are isolated, and performing arbitrary path merging and path unbalancing operations (some of which are preceded by either one or two rotations) until no longer possible.

Given a maximal path cover $(P_1, P_2, \ldots, P_p)$ of $H$ (where $p$ denotes the number of paths in the path cover), suppose that paths are sorted in order of increasing lengths, breaking ties arbitrarily. Hence $1 \le |P_1| \le |P_2| \le \ldots \le |P_p|$. We consider the following classes of vertices of $W$:

\begin{enumerate}

\item {\em Isolated vertices} $Q$. These are the vertices that belong to paths of length~1. Let $k$ denote the number of isolated vertices, and let us name them as $q_1, \dots, q_k$. Observe that $|P_k| = 1$ and $|P_{k+1}| > 1$.

\item {\em Start vertices} $S$. These are the starting vertices of those paths that have length larger than~1. The start vertex of path $j$, for $k < j \le p$, is denoted by $s_j$.

\item {\em End vertices} $T$. These are the end vertices of those paths that have length larger than~1. The end vertex of path $j$, for $k < j \le p$, is denoted by $t_j$.

\end{enumerate}

\begin{lemma}
\label{lem:sort1}
If $H$ has a maximal path cover with $p$ paths and $k$ isolated vertices, then $\rho(G) \ge \frac{2p - k}{n}$.
\end{lemma}

\begin{proof}
Maximality of the path cover implies the following regarding arcs in $H$:

\begin{enumerate}

\item There is no arc from $T$ to $Q \cup S$, except for arcs of the form $(t_j,s_j)$.

\item There is no arc from $t_j$ to $t_i$ if $j > i$.

\item There is no arc from $Q$ to $S$.

\item There is no arc from $q_i$ to $q_j$ for $i \not= j$.

\item There is no arc from $s_i$ to $s_j$ if $i < j$.

\end{enumerate}

Consider the following order on the vertices of $Q \cup S \cup T \in W$:
$$s_p, \ldots, s_{k+1}, q_1, \ldots q_k, t_{k+1}, \ldots, t_p.$$

In the above order, the only arcs of $H$ that go from a later vertex to an earlier vertex are of the form $(t_j,s_j)$ (for a path $P_j$ that can undergo a rotation). Consider in the graph $G$ a permutation $\pi$ on $V$ that starts with the vertices in $V$ in their order of appearance in the above sorted order. We claim that regardless of $\sigma$, all the prefix will be matched. As the length of this prefix is $2p - k$, this will prove the lemma.

It remains to prove the claim. Suppose first that in the above order there are no arcs of $H$ that go from a later vertex to an earlier vertex. Then earlier vertices
in this prefix cannot be spoiling vertices for later vertices. Hence indeed,
regardless of $\sigma$, all the prefix will be matched.

Suppose now that in the above order there are arcs of $H$ that go from a later vertex to an earlier vertex. As noted above, such an arc would be of the form $(t_j,s_j)$. We need to show that even if $s_j$ acts as a spoiling vertex for $t_j$ under $\pi$ and $\sigma$, the vertex of $V$ corresponding to $t_j$ (recall that $t_j$ is a vertex of $W$) will still be matched. Consider the path $P_j$, and let us rename its vertices as $w_1, \ldots, w_{\ell}$ (where previously we used $s_j = w_1$ and $t_j = w_{\ell}$). Recalling that $w_i = (u_i,v_i)$, we wish to show the $v_{\ell}$ would be matched even if $u_{\ell}$ is matched to $v_1$. The path $P_j$ implies that $u_{\ell - 1}$ is a neighbor of $v_{\ell}$ in $G$. Hence $v_{\ell}$ will be matched if no vertex preceding $v_{\ell}$ in $\pi$ is matched to $u_{\ell - 1}$. As $P_j$ could be rotated to $w_{\ell}, w_1, \ldots, w_{\ell-1}$ and the path cover was maximal, we get that there is no arc in $H$ from $w_{\ell - 1}$ to any of the vertices $s_p, \ldots, s_{k+1}, q_1, \ldots q_k, t_{k+1}, \ldots, t_{j-1}$ (except that there might be an arc from $w_{\ell-1}$ to $s_j = w_1 = (u_1,v_1)$, but $v_1$ was assumed to be matched to $u_{\ell}$ and hence cannot be matched to $u_{\ell - 1}$), hence none of the corresponding vertices in $V$ can be matched to $u_{\ell - 1}$.
\end{proof}

\begin{lemma}
\label{lem:sort2}
If $H$ has a maximal path cover with $p$ paths, then $\rho(G) \ge \frac{5}{9} - \frac{p}{9n}$.
\end{lemma}

\begin{proof}

In every path $P$ in the path cover, we can partition the set of its vertices into {\em odd} and {\em even} in an alternating fashion (walking along the path vertices alternate as being odd or even).
Dropping the first vertex from each path, the remaining vertices can be partitioned into two sets, $W_1$ and $W_2$, where $W_1$ is the set of odd vertices and $W_2$ is the set of even vertices. 
Similarly in $G$, we define the odd set $V_1$ by $v_i \in V_1$ iff $w_i \in W_1$, and the even set $V_2$ by $v_i \in V_2$ iff $w_i \in W_2$.

With every vertex $v_j \in V_1$ we associate two {\em designated} neighbors in $U$, not shared by any other vertex in $V_1$. One of its designated neighbors is $u_j$. The other is determined from the path that contains $w_j$. Let the incoming arc into $w_j$ in this path be $(w_i,w_j)$. Then the other designated neighbor of $v_j$ is $u_i$. Note that $v_i \not\in V_1$, and hence the designated neighbors of $v_j$ are distinct from the designated neighbors of any other vertex in $V_1$. In an analogous fashion, we associate two designated neighbors with each vertex in $V_2$, distinct from designated neighbors of other vertices in $V_2$.

Observe that $|V_1| + |V_2| = n - p$. We may assume without loss of generality that $|V_1|\leq|V_2|$ and $|V_2|-|V_1| \le 1$, because in each path the number of even vertices differs from the number of odd vertices by at most one, and when creating $V_1$ and $V_2$ path by path, we can add the larger set of vertices from the path to the smaller of the sets $V_1$ or $V_2$.
Hence it holds that $|V_1| = \lfloor \frac{n-p}{2} \rfloor$ and $|V_2| = \lceil \frac{n-p}{2} \rceil$.

Now we construct a permutation $\pi$ over $V$. The prefix of $\pi$ is composed of start vertices of the paths (the sets $S$ and $Q$), and is a prefix of the prefix given in Lemma~\ref{lem:sort1}. Namely, it is $s_p, \ldots, s_{k+1}, q_1, \ldots q_k$. Thereafter $\pi$ continues with $V_1$ followed by $V_2$.

Regardless of $\sigma$, all $p$ vertices of $S$ and $Q$ are matched, as in Lemma~\ref{lem:sort1}. For a given $\sigma$, let $n_1$ be the number of vertices matched in $V_1$ and let  $n_2$ be the number of vertices matched in $V_2$.
Then, $\lfloor \frac{n-p}{2} \rfloor - n_1$, the number of unmatched vertices  in $V_1$, satisfies $2(\lfloor \frac{n-p}{2} \rfloor - n_1) \le p + n_1$, because each unmatched vertex in $V_1$ has two distinct designated neighbors in $U$ that need to be matched to earlier vertices in $S \cup Q \cup V_1$.
Likewise,  $\lceil \frac{n-p}{2} \rceil - n_2$, the number of unmatched vertices  in $V_2$, satisfies $2(\lceil \frac{n-p}{2} \rceil - n_2) \le p + n_1 + n_2$.
Adding two times the first inequality and three times the second we get that $5p + 5n_1 + 3n_2 \ge 5n - 5p - 4n_1 - 6n_2$, implying that $p + n_1 + n_2 \ge \frac{5n}{9} - \frac{p}{9}$, as desired.
\end{proof}

Note that Lemma~\ref{lem:sort2} directly implies that $\rho(G) \geq \frac{5}{9} - \frac{1}{9n}$ for every Hamiltonian graph.



The following proposition applies to any prefix under $\pi$.

\begin{proposition}[Prefix matching]
\label{prop:prefix}
Fix a perfect matching $M$. For every integer $k$, let $V_{[k]}$ denote the first $k$ vertices under $\pi$, and let $U_{[k]}$ denote their partners under $M$.
For every $\sigma, \pi$ and every $k$, suppose $\ell$ vertices in $V_{[k]}$ are matched in $M[\sigma,\pi]$ to vertices in $U \setminus U_{[k]}$.
Then, at least $\ell + (k-\ell)/2$ vertices in $V_k$ are matched in $M[\sigma,\pi]$.
\end{proposition}

\begin{proof}
Let $L$ denote the set of $\ell$ vertices in $V_{[k]}$ matched to vertices in $U \setminus U_{[k]}$.
For every unmatched $v$ in $V_{[k]} \setminus L$, $v$'s partner has to be matched to a unique vertex $v' \in V_{[k]} \setminus L$ such that $\pi[v']<\pi[v]\leq k$.
It follows that the number of unmatched vertices in $V_k \setminus L$ is at most the number of matched vertices in this set, concluding the proof.
\end{proof}

Let $W_1$ denote the vertices in $H$ that belong to paths of length~1 in the path cover, and let $W_2$ denote the remaining vertices (belonging to longer paths). Likewise, let $V_1$ and $V_2$ denote the corresponding sets of vertices in $V$, and let $U_1$ and $U_2$ denote the corresponding sets of vertices in $U$.

Considering only the arcs leading from $W_1$ to $W_2$, let $M_{12}$ denote the maximum matching among these arcs.

\begin{lemma}
\label{lem:M12}
For $G$ as above, $\rho(G)n \ge |W_1|+\frac{|W_2|}{2} - \frac{|M_{12}|}{2}$.
\end{lemma}

\begin{proof}
Consider a permutation $\pi$ over $V$ in which $V_2$ precedes $V_1$. Let $\sigma$ be an arbitrary permutation over $U$. Let $m$ denote the number of vertices in $V_2$ that are matched to vertices in $U_1$.
Then $m \le |M_{12}|$. By Proposition~\ref{prop:prefix}, the total number of vertices matched in $V_2$ is at least $m + \frac{1}{2}(|V_2| - m)$. In addition, the total number of vertices matched in $V_1$ is at least $|V_1| - m$, by the fact that every vertex in $W_1$ belongs to a path in $H$ of length 1. Hence the total number of vertices matched is at least $|W_1|+\frac{|W_2|}{2}- \frac{m}{2} \ge |W_1|+\frac{|W_2|}{2} - \frac{|M_{12}|}{2}$, as desired.
\end{proof}

%

For a maximal path cover with $p$ paths and $k$ isolated vertices, we define $\epsilon_1 = \frac{1}{2} - \frac{k}{n}$.
We also define $\epsilon_3 = \frac{1}{2} - \frac{1}{n}|M_{12}|$.

\begin{lemma}
\label{lem:largeM12}
With parameters defined as above, $\rho(G) \ge \frac{2}{3} - \frac{1}{3}(\epsilon_1 + \epsilon_3)$.
\end{lemma}

\begin{proof}
According to the parameters defined above, $|M_{12}|=n(\frac{1}{2}-\epsilon_3),$
and $|W_1|=k=n(\frac{1}{2}-\epsilon_1)$ \big(and hence $|W_2|=n(\frac{1}{2}+\epsilon_1)$\big).
Consider a permutation $\pi$ over $V$ in which $V_1$ precedes $V_2$. Let $\sigma$ be an arbitrary permutation over $U$. All vertices in $V_1$ are matched since every vertex in $W_1$ belongs to a path in $H$ of length 1. As to the vertices in $V_2$, partition then into two sets as follows. $V_2^+$ is the set of vertices participating in the matching $M_{12}$ (hence $|V_2^+| = |M_{12}|$), and $V_2^-$ are the remaining vertices (hence $|V_2^-| = |W_2| - |M_{12}|)$. Now we assign distinct designated neighbors to vertices in $V_2$. Each vertex in $V_2^-$ has one designated neighbor (its partner in $U$), whereas each vertex in $V_2^+$ has two designated neighbors (its partner in $U$ and one more neighbor through the matching $M_{12}$). Let $x$ denote the number of vertices in $V_2$ matched under $(\pi,\sigma)$.
Then $|W_2|-x$ 
vertices of $V_2$ remain unmatched and their designated neighbors need to be matched. 

We claim that the total number of matched vertices satisfies $|W_1| + x \ge |W_2| - |M_{12}| + 2(|M_{12}| - x)$.
To see this, if $x \leq |M_{12}|$, then the right hand side is a lower bound on the number of designated neighbours of unmatched vertices in $V_2$.
If $x > |M_{12}|$, then we can use $|W_1| + x \geq |W_2|-x \geq |W_2| - |M_{12}| + 2(|M_{12}| - x)$, where the first inequality holds since the right hand side is a lower bound on the number of designated neighbours of unmatched vertices in $V_2$.

We get that
\begin{eqnarray*}
	x&\geq& \frac{|W_2|+|M_{12}|-|W_1|}{3}\\
	& = & \frac{n(\frac{1}{2}+\epsilon_1)+n(\frac{1}{2}-\epsilon_3)-n(\frac{1}{2}-\epsilon_1)}{3}\\
	& = & \left(\frac{1}{6} + \frac{2\epsilon_1}{3} - \frac{\epsilon_3}{3}\right)n.
\end{eqnarray*}
Therefore, the size of the matching is at least
$$
|W_1| + x \geq \left( \frac{1}{2} - \epsilon_1\right)n + \left(\frac{1}{6} + \frac{2\epsilon_1}{3} - \frac{\epsilon_3}{3}\right)n = \left(\frac{2}{3} - \frac{1}{3}(\epsilon_1 + \epsilon_3)\right) n.
$$
\end{proof}

\begin{thm}
\label{thm:main}
There is some absolute constant $\rho \ge \frac{1}{2} + \frac{1}{86}$ such that $\rho(G) \ge \rho$ for every graph $G$. Moreover, a permutation $\pi$ ensuring a matching of size at least $\rho n$ can be computed in polynomial time.
\end{thm}

\begin{proof}
Recall that for a maximal path cover with $p$ paths and $k$ isolated vertices, $\epsilon_1 = \frac{1}{2} - \frac{k}{n}$. We also define $\epsilon_2 = \frac{p-k}{n}$. Hence $p = (\frac{1}{2} - \epsilon_1 + \epsilon_2)n$. Observe that $\epsilon_2 \ge 0$, though $\epsilon_1$  might be negative. By Lemma~\ref{lem:sort1} we have:

$$\rho(G) \ge \frac{1}{2} - \epsilon_1 + 2\epsilon_2.$$

By Lemma~\ref{lem:sort2} we have:

$$\rho(G) \ge \frac{1}{2} + \frac{\epsilon_1}{9}  - \frac{\epsilon_2}{9}.$$

Let $W_1$ (of size $(\frac{1}{2} - \epsilon_1)n$) and $W_2$ (of size $(\frac{1}{2} + \epsilon_1)n$) and $M_{12}$ be as defined prior to Lemma~\ref{lem:M12}.
Recall that $\epsilon_3 = \frac{1}{2} - \frac{1}{n}|M_{12}|$. By Lemma~\ref{lem:M12} we have:

$$\rho(G) \ge \frac{1}{2} - \frac{\epsilon_1}{2}  + \frac{\epsilon_3}{2}.$$

By Lemma~\ref{lem:largeM12} we have:

$$\rho(G) \ge \frac{2}{3} - \frac{1}{3}(\epsilon_1 + \epsilon_3).$$

Now the optimal value of $\rho$ can be found by minimizing $\rho$ subject to the constraints $\epsilon_2, \epsilon_3 \ge 0$ and the four inequalities derived above for $\rho$.

The first two lemmas imply that $\rho(G) \ge \frac{1}{2} + \frac{\epsilon_1}{19}$, and the last two lemmas imply that $\rho(G) \ge \frac{3}{5} - \frac{2\epsilon_1}{5}$. For $\epsilon_1 = \frac{19}{86}$ both inequalities give $\rho(G) \ge \frac{1}{2} + \frac{1}{86}$, and for other values of $\epsilon_1$ at least one of the two inequalities gives an even better bound. 
For $\epsilon_1 = \frac{19}{86}$, $\epsilon_2 = \frac{10}{86}$ and $\epsilon_3 = \frac{21}{86}$ all four lemmas give $\rho(G) \ge \frac{1}{2} + \frac{1}{86}$, and for other values of $\epsilon_1, \epsilon_2, \epsilon_3$ at least one of the four lemmas gives an even better bound. 

All steps in the analysis above can be implemented by polynomial time algorithms.
\end{proof}

%% file: regular.tex
\section{Regular Graphs}
\label{sec:regular}

In this section we consider the case where $G(U,V;E)$ is a $d$-regular bipartite graph with $2n$ vertices. Given that such graphs have $d$ edge disjoint perfect matchings, one can hope to achieve higher values for $\rho$ for these graphs.

\subsection{Positive Result}
\label{sec:regular-pos}
The following proposition establishes a lower bound on $\rho$, as a function of $d$.

\begin{proposition}
\label{pro:smalld}
For every $d$-regular graph $G$, it holds that $\rho[G] \ge  \frac{d}{2d-1}$.
\end{proposition}

\begin{proof}
Since the greedy algorithm produces a maximal matching, it suffices to show that every maximal matching in a $d$-regular graph has size at least $\frac{d}{2d-1}n$.
To see this, let $S \subset U$ and $T \subset V$ be the sets of unmatched nodes in an arbitrary maximal matching, and suppose $|S|=|T|=(1-\alpha) n$. The nodes in $S,T$ must form an independent set.
Consider the size of the edge set connecting $S$ and $V \setminus T$.
On the one hand, this size equals $(1-\alpha) n d$ (since all edges from $S$ go to $V \setminus T$); on the other hand, this size is at most $\alpha n (d-1)$ (since at least one edge from each node in $V \setminus T$ goes to $U \setminus S$).
Thus, $(1-\alpha) n d \leq \alpha n (d-1)$, implying that $\alpha \geq d/(2d-1)$.
Hence we have that $|M_G[\sigma,\pi]| \ge  \frac{d}{2d-1}n$, for every $\pi$.
\end{proof}

Remark: For every $d$ there exists a $d$-regular graph with a perfect matching that admits a maximal matching of size $\frac{d}{2d-1}n$. Suppose that $n = 2d-1$, and
consider a $d$-regular graph where $|S|=|T|=d-1$ for some $S \subset U, T \subset V$, every node in $U \setminus S$ is connected to a single, different node in $V \setminus T$, and to all $d-1$ nodes in $T$, and every node in $V \setminus T$ is connected to a single, different node in $U \setminus S$, and to all $d-1$ nodes in $S$. The perfect matching between $U \setminus S$ and $V \setminus T$ is a maximal matching of size $\frac{d}{2d-1}n$.

The lower bound of Proposition~\ref{pro:smalld} approaches $\frac{1}{2}$ from above as $d$ grows.
The following theorem shows that there exists some permutation $\pi$ that ensures that the fraction of matched vertices approaches $5/9$.
This is a direct corollary from Lemma \ref{lem:sort2} and a theorem in \cite{FRS14}.

\begin{corollary}
\label{cor:regular}
For $d$-regulars bipartite graphs, $\rho \ge \frac{5}{9} - O(\frac{1}{\sqrt{d}})$.
\end{corollary}

\begin{proof}
Theorem~3 in~\cite{FRS14} shows that every $n$-vertex $d$-regular graph has a path cover (referred to as a {\em linear forest}) with  $p = O(\frac{n}{\sqrt{d}})$ paths.  By Lemma~\ref{lem:sort2}, $\rho(G) \ge \frac{5}{9} - O(\frac{1}{\sqrt{d}})$.
\end{proof}

{\bf Remarks:}

\begin{enumerate}

\item For small $d$, the bound of $\rho \ge \frac{d}{2d - 1}$ which holds for every maximal matching is stronger than the bound in Corollary~\ref{cor:regular}.

\item The proof of Corollary~\ref{cor:regular} extends to graphs that are nearly $d$-regular, by using Theorem~5 from~\cite{FRS14}.

\item For $d$-regular graphs, conjectures mentioned in~\cite{FRS14} combined with our proof approach suggest that $\rho \ge \frac{5}{9} - O(\frac{1}{d})$.

\end{enumerate}

\subsection{Negative Result}
\label{sec:regular-neg}

The following example shows that even in a regular graph with arbitrarily high degree, there may be no permutation $\pi$ that ensures to match more than a fraction $8/9$ of the vertices.

\begin{thm}
For every integers $d, t \ge 1$, there is a regular bipartite graph $G_{d,t}$ of even degree $2d$ and $n= 3dt$ vertices on each side such that $\rho(G_{d,t}) \le \frac{8}{9}$.
\end{thm}

\begin{proof}
Consider a regular bipartite graph $G(U,V;E)$ with even degree $2d$, and $3d$ vertices on each side. To define the edge set, let $U = U_1 \cup U_2 \cup U_3$ with each $U_i$ of cardinality $d$, and similarly $V = V_1 \cup V_2 \cup V_3$ with each $V_i$ of cardinality $d$. For every $i \not= j$, we have a complete bipartite graph between $U_i$ and $V_j$, and for every $i$, there are no edges between $U_i$ and $V_i$.

Let $\pi$ be an arbitrary permutation over $V$, let $S$ be the first $2d$ vertices in $\pi$, and let $T$ be the last $d$ vertices. Let $i$ be such that $|V_i \cap T|$ is largest (breaking ties arbitrarily). Without loss of generality we may assume that $i = 3$, and then $|V_3 \cap T| \ge d/3$. Hall's condition implies that there is a perfect matching between $U_1 \cup U_2$ and $S$ (and more generally, between $U_1 \cup U_2$ and any $2d$ vertices from $V$). Hence one can choose a permutation $\sigma$ over $U$ whose first $2d$ vertices are $U_1 \cup U_2$ that will match the vertices of $S$ one by one. Thereafter, the vertices of $T \cap V_3$ will remain unmatched.

To get the graph $G_{d,t}$ claimed in the theorem, take $t$ disjoint copies of $G(U,V;E)$ above.
\end{proof}

%% file: random.tex
\section{Random Permutation}
\label{sec:random}

In this section we consider scenarios in which the designer is unaware of the graph structure.
In such scenarios, the best she can do is impose a random permutation over the vertices in $V$.
Thus, we study the performance of a random permutation.

We first show that there exists a graph $G$ for which a random permutation does not match significantly more than a half of the vertices, even if every vertex has a high degree.

\begin{proposition}
There exists a bipartite graph $G(U,V;E)$ such that a random permutation gets $\rho(G)=\frac{1}{2}+o(1)$ almost surely.
\label{ex:half}
\end{proposition}

\begin{proof}
Consider the graph $G(U,V;E)$, where $U=(U_1,U_2)$, $V=(V_1,V_2)$, and each of $U_1,U_2,V_1,V_2$ is of size $n/2$. The set of edges constitutes of a perfect matching between $U_1$ and $V_1$, a perfect matching between $U_2$ and $V_2$, and a bi-clique between $U_1$ and $V_2$.
Let $\pi$ be a random permutation.
With high probability, for each vertex $v_1 \in V_1$, except for $\sim\sqrt{n}$ such vertices, we can associate a unique vertex $v_2 \in V_2$ that precedes $v_1$ in $\pi$. Let $S \subseteq V_1$ denote this set.
Consider an arrival order $\sigma$ in which agents in $U_1$ arrive first, with a vertex $u_{1j}$ preceding a vertex $u_{1j'}$ if $\pi(v_{2j})<\pi(v_{2j'})$. Every vertex in $U_1$ such that its neighbor in $V_1$ (according to the perfect matching) belongs to $S$ will be matched to the corresponding vertex in $V_2$.
Therefore, all but $\sim\sqrt{n}$ vertices of $V_1$ remain unmatched, and the size of the matching is $n(1/2 +o(1))$, whereas $OPT=n$.
\end{proof}

In the above example, if the degree 1 vertices of $V$ are placed in the prefix of $\pi$, then the matching obtained is optimal. Hence, one might think that if we prioritize low degree neighbors in $\pi$, and randomize w.r.t. the partial order, we might get a good approximation.
However, one can make a similar example to the one above, where vertices are partitioned into sets of perfect matchings of size $\sqrt{n}$, $\{(U_{11},V_{11}),\ldots, (U_{1\sqrt{n}},V_{1\sqrt{n}}), (U_{21},V_{21}),\ldots, (U_{2\sqrt{n}},V_{2\sqrt{n}})\}$. Each $V_{1i}$ is also connected in a bi-clique to $U_{2i}$, and in addition, there are sets $U',V'$ of size $\sqrt{n}$ connected to the vertices of the other side to balance out the degrees. A similar argument shows that in this graph, a random permutation performs badly as well (no use of prioritizing low degree vertices since the degrees of all vertices are the same).

In contrast to the last examples, in some classes of graphs, a random permutation guarantees to match a fraction of the vertices that is bounded away from a half.
This is the case, for example, in hamiltonian graphs.
The formal statement and proof are deferred to Section~\ref{sec:hamiltonian}.

%% file: hamiltonian.tex
\section{Hamiltonian Bipartite Graphs}
\label{sec:hamiltonian}

In this section we show that in hamiltonian graphs the fraction of vertices that is guaranteed to be matched in a random permutation is bounded away from a half. This is in contrast to general graphs (see Section~\ref{sec:random}).

\begin{thm}
	For every Hamiltonian graph $G$, it holds that $E_{\pi}[\min_{\sigma}[|M_G[\sigma,\pi]|]] > 0.5012$.
\label{thm:random-hamiltonian}
\end{thm}	

\subsection{Proof Approach}
\label{sec:proof-overview}

We first provide a high level overview of our proof approach.

A permutation $\pi$ (over $V$) is said to be {\em safe} for a set $S \subset V$ if for every permutation $\sigma$ (over $U$) the greedy process matches at least one vertex in $S$ (i.e., no $\sigma$ leaves all vertices in $S$ unmatched).
Fix some constant $\epsilon$. In order to establish that $\rho \geq (1/2 + \epsilon)$, we need to show that there exists a permutation $\pi$ that is safe for every set $S$ of size $(1/2-\epsilon)n$.
Our proof approach is the following: we show that for a permutation $\pi$ chosen uniformly at random, the expected number (expectation taken over choice of $\pi$) of sets of size $(1/2-\epsilon)n$ for which $pi$ is unsafe is smaller than $1$.
This implies that there exists a permutation $\pi$ that is safe for all sets of size $(1/2-\epsilon)n$, as desired.

First, we define a collection of sets that can potentially remain unmatched (``bad" sets).
Let $B_\epsilon$ denote the set of all sets $S \subset U$ of size $(1/2-\epsilon)n$ such that there exists a permutation $\pi$ that is unsafe for $S$.

Second, for a given set $S$ and permutation $\pi$ we identify a sufficient condition for $\pi$ to be safe for $S$.
Let $S' \subset S$ be the lowest $\alpha n$ vertices in $S$ (according to $\pi$), let $v'$ be the last vertex in $S'$ (i.e., the vertex with rank $\alpha n$ in $S'$), and let $P$ be the set of vertices in $V - S'$ that precede $v'$ in $\pi$.
We claim that if the size of $P$ is smaller than the size of $N(S')$ (the neighbors of $S'$), then $\pi$ is safe for $S$.
To see this, assume by way of contradiction that $\pi$ is unsafe for $S'$.
This implies that every vertex in $N(S')$ is matched to a vertex in $V - S'$.
Since there are strictly less than $|N(S')|$ vertices in $V - S'$ that precede $v'$, at least one of the vertices in $N(S')$ must be matched to a vertex higher than $v'$. But, this vertex has a neighbor in $S'$ with lower rank, contradicting the greedy process.

We now proceed by establishing the following three lemmas:
\begin{itemize}
  \item {Few bad sets lemma:} the size of $B_\epsilon$ is at most $n_B=n_B(\epsilon)$.
  \item {\em Expansion lemma:}
  given a set $S \subset V$ and parameters $\alpha, \beta$, the probability (over a random choice of $\pi$) that the lowest $\alpha n$ vertices in $S$ have less than $\beta n$ neighbors is at most $p=p(\alpha,\beta)$.
  \item {\em Good order lemma:} given a set $S \subset V$ and parameters $\alpha, \beta$, the probability (over a random choice of $\pi$) that the $(\alpha n)^{th}$ lowest vertex in $S$ is higher than $\beta n$ vertices in $V \setminus S$ is at most $q=q(\alpha,\beta)$.
\end{itemize}

The three lemmas are combined as follows.
For a given set $S$, due to the sufficient condition identified above, it follows from the union bound that the probability that a uniformly random permutation $\pi$ is unsafe for $S$ is at most $p+q$.
Applying the union bound once more over all bad sets (at most $n_B$ sets, as implied by the few bad sets lemma), implies that the probability that a uniformly random permutation $\pi$ is unsafe for some set of size $(1/2-\epsilon)n$ is at most $n_B(p+q)$.
Thus, to conclude the proof, it remains to find parameters such that $n_B(p+q) < 1$.

The good order lemma is independent of the graph structure.
In contrast, the expansion lemma and the few bad sets lemma rely heavily on the structure of the graph.
As it turns out, Hamiltonian graphs have properties that enable us to establish the two lemmas with good parameters.


\subsection{Formal Proof}

Throughout this section we use $H(\cdot)$ to denote the binary entropy function; i.e., given a constant $p\in (0,1)$, $H(p)=-p\log_2 p-(1-p)\log_2(1-p)$.

\begin{fact}[Stirling's Approximation]
	As $n\rightarrow\infty$, $$n!=(1+o(1))\sqrt{2\pi n}\left(\frac{n}{e}\right)^n.$$
\end{fact}
Using Stirling's Approximation, one can derive the following bound.
\begin{fact}\label{fact:binom_estimate}
	For $n$ and $k=pn$ for some constant $p\in (0,1)$,
	\begin{eqnarray}
		\binom{n}{k}=2^{(H(p)+o(1))n},\label{eq:binom_approx}
	\end{eqnarray}
	where $H(p)=-p\log_2 p-(1-p)\log_2(1-p)$ is the binary entropy function.
\end{fact}

We first establish the good order lemma. This lemma is independent of the graph structure.

\begin{lemma} [Good order lemma]
	Let $\alpha < \beta < 1$, $\rho= \frac{1}{2}+\epsilon$ for some $\epsilon>0$ 
	and $\bar{\rho}=1-\rho$ such that $\frac{\beta}{\alpha} > \frac{\rho}{\bar{\rho}}$. Let $S\subset V$ be a set of size $\bar{\rho}n$. The probability that in a random permutation $\pi$ there are at least $\beta n$ vertices of $V\setminus S$  before $\alpha n$ vertices from $S$ is at most $2^{-(H(\alpha+\beta)-H(\frac{\alpha}{\bar{\rho}})\bar{\rho}-H(\frac{\beta}{{\rho}}){\rho}-o(1))n}$.\label{lem:order}
\end{lemma}
\begin{proof}
	We first analyze the case that in the first $(\alpha+\beta)n$ vertices in $\pi$ there are {\em exactly} $\alpha n$ vertices from $S$. The number of possibilities for this case is $\binom{\bar{\rho} n}{\alpha n}\binom{\rho n}{\beta n}$.
	
	Let $\beta' = \beta+x$ and $\alpha'=\alpha-x$. By the conditions on $\alpha,\beta$ and $\epsilon$, we have that $\frac{\beta'}{\alpha'}\geq\frac{\beta}{\alpha}\geq \frac{\rho}{\bar{\rho}}$. Therefore,
	\begin{eqnarray*}
		& &  \beta'\bar{\rho}\geq \alpha'\rho\Rightarrow \beta'\bar{\rho}-\alpha'\beta'\geq \alpha'\rho-\alpha'\beta'\Rightarrow \frac{(\bar{\rho}-\alpha')}{\alpha'}\cdot \frac{\beta'}{(\rho-\beta')}\geq 1\\
		& & \Rightarrow \frac{(\bar{\rho}n-\alpha'n+1)}{\alpha'n}\cdot \frac{\beta'n+1}{(\rho n-\beta'n)}\geq 1 \iff\frac{\binom{\bar{\rho} n}{\alpha' n}\binom{\rho n}{\beta' n}}{\binom{\bar{\rho} n}{\alpha' n-1}\binom{\rho n}{\beta' n+1}}>1.
	\end{eqnarray*}
	It follows that
	$\binom{\bar{\rho}n}{\alpha n}\binom{\rho n}{\beta n}> \binom{\bar{\rho} n}{\alpha'n}\binom{\rho n}{\beta' n}$ for every $\alpha'< \alpha$  and $\beta'>\beta$ such that $\alpha+\beta=\alpha'+\beta'$. Therefore, the probability to have at most $\alpha n$ vertices from $S$ in the first $(\alpha+\beta)n$ vertices in $\pi$ is at most
	\begin{eqnarray*}
		\frac{\alpha n\cdot \binom{\bar{\rho}n}{\alpha n}\binom{\rho n}{\beta n}}{\binom{n}{(\alpha+\beta)n}}&=&\frac{2^{(H(\frac{\alpha}{\bar{\rho}})\bar{\rho}+H(\frac{\beta}{{\rho}}){\rho}+o(1))n}}{2^{(H(\alpha+\beta)+o(1))n}}\\
		& = & 2^{-(H(\alpha+\beta)-H(\frac{\alpha}{\bar{\rho}})\bar{\rho}-H(\frac{\beta}{{\rho}}){\rho}-o(1))n},
	\end{eqnarray*}
	where the first equality follows Fact \ref{fact:binom_estimate}.
\end{proof}

Let $\rho= \frac{1}{2}+\epsilon$ for some constant $\epsilon>0$, $\bar{\rho}=1-\rho=\frac{1}{2}-\epsilon$. The next lemma will be used in order to prove the few bad sets lemma and the expansion lemma. It uses the existence an Hamiltonian cycle in the graph in order to claim that most sets will have a large number of neighbors. Therefore, a random set will have a large expansion. In addition, there will be few sets of size $(\frac{1}{2}-\epsilon)n$ with less than $(\frac{1}{2}+\epsilon)n$ neighbors (i.e., a few bad sets).

\begin{lemma}
	Let $\alpha\in (0, 1/2)$ and $\beta \in (\alpha, 1)$ be two constants such that $\delta=\beta-\alpha < \alpha/2$. The number of sets $S$ of size $\alpha n$ where $|N(S)|\leq \beta n$ is at most $2^{(\alpha H(\frac{\delta}{ \alpha})+(1-\alpha)H(\frac{\delta} {(1-\alpha)})+o(1))n}.$
\label{lem:ham_alpha_beta}
\end{lemma}
\begin{proof}
\begin{AvoidOverfullParagraph}
	Consider a Hamiltonian cycle that traverses through the graph's vertices $H=(v_1,u_1,v_2,u_2,\ldots, v_n,u_n,v_1)$, where $\{v_i\}_{i\in [n]}=V$ and $\{u_i\}_{i\in [n]}=U$. Let $S$ be some set of vertices from $V$ of cardinality $\rho n$. Note that in the cycle $H$, each vertex $v$ of $S$ has two neighbors, where one of these neighbors is joined with  an adjacent vertex from $V$ in the cycle. Therefore, the number of neighbors of a sequence of $k$ consecutive vertices of $V$ in $H$ is $k+1$. Thus, the set $N(S)$ is of size $\alpha n$ plus the number of consecutive blocks of vertices from $V$ chosen.
\end{AvoidOverfullParagraph}
	
	We bound the number of ways to pick at most $\delta n$ consecutive blocks of vertices from $V$. We first bound the number of ways to pick {\em exactly} $\delta n$ such blocks.
	In this case, the $\alpha n$ chosen elements have to be within $\delta n$ blocks. The number of ways to partition $\alpha n$ elements to $\delta n$ noen empty blocks is $\binom{\alpha n - 1}{\delta n - 1}$. After deciding the number of elements in each block, we need to figure out their  location along the Hamiltonian cycle. $(1-\alpha)n$ elements reside outside of the blocks of chosen $\alpha n$ elements. We need to chose the location of the first block in $H$ (for which there are $n$ possibilities), and then the number of element between each block, where two blocks are separated by at least one element. The latter is equivalent to splitting $(1-\alpha)n$ elements into $\delta n$ non empty bins, for which there are $\binom{(1-\alpha)n-1}{\delta n - 1}$ possibilities. Overall, there are $n\binom{\alpha n - 1}{\delta n - 1}\binom{(1-\alpha)n-1}{\delta n - 1}$ such possibilities\footnote{Notice there's some over-counting in this argument, but this bound suffices for our purpose.}.
	
	For $\delta'< \delta$, one can similarly devise the bound of $n\binom{\alpha n - 1}{\delta' n - 1}\binom{(1-\alpha)n-1}{\delta' n - 1}$ which is smaller than $n\binom{\alpha n - 1}{\delta n - 1}\binom{(1-\alpha)n-1}{\delta n - 1}$ by our conditions on $\alpha$ and $\delta$. Overall, we can bound the number of ways to pick at most $\delta n$ consecutive blocks of vertices from $V$ by
	\begin{eqnarray*}
	\delta n^2\binom{\alpha n - 1}{\delta n - 1}\binom{(1-\alpha)n-1}{\delta n - 1} & < & \delta n^2\binom{\alpha n}{\delta n}\binom{(1-\alpha)n}{\delta n}\\
	&=& 2^{o(1)n}\cdot 2^{(H(\frac{\delta}{ \alpha})+o(1))\alpha n}\cdot 2^{(H(\frac{\delta} {(1-\alpha)})+o(1))(1-\alpha) n}\\
	& = &2^{(\alpha H(\frac{\delta}{ \alpha})+(1-\alpha)H(\frac{\delta} {(1-\alpha)})+o(1))n},
	\end{eqnarray*}
	where the first equality follows Fact \ref{fact:binom_estimate}.
\end{proof}

The expansion and few bad sets lemmas are obtained as direct corollaries of Lemma~\ref{lem:ham_alpha_beta}.

\begin{lemma}[Few bad sets Lemma for Hamiltonian graphs]
	Let $\epsilon$ be a constant such that $\epsilon<0.1$. The number of bad sets in any Hamiltonian graph is at most
	$$|B_\epsilon|\leq  2^{\left(\bar{\rho} H(\frac{2\epsilon}{\bar{\rho}})+\rho H(\frac{2\epsilon} {\rho})+o(1)\right)n}.$$ \label{lem:ham_bad_sets}
\end{lemma}
\begin{proof}
	Notice that if a set $S$ of size $\bar{\rho} n = (\frac{1}{2}-\epsilon)n$ has more than $\rho n$ neighbors, it cannot be left unmatched, since at least one of it's neighbors will could not be matched to $V\setminus S$.
	A direct application of Lemma~\ref{lem:ham_alpha_beta} yields that the number of such sets is at most $2^{\left(\bar{\rho} H(\frac{2\epsilon}{\bar{\rho}})+\rho H(\frac{2\epsilon} {\rho})+o(1)\right)n}.$
\end{proof}

We note that this lemma is not true for general graphs.
An example of a graph that admits $2^{n/4}$ bad sets is given in Proposition~\ref{ex:many-bad}.

\begin{lemma}[Expansion Lemma for Hamiltonian graphs]
	Consider a set $S\subset V$ of size $\bar{\rho} n$ and parameters $\alpha$, $\beta$. The probability that the lowest $\alpha n$ vertices in $S$ have less than $\beta n$ neighbors is at most
	$$2^{\left(-H(\frac{\alpha}{\bar{\rho}})+\alpha H(\frac{\delta}{ \alpha})+(1-\alpha)H(\frac{\delta} {(1-\alpha)})+o(1)\right)\bar{\rho}n}.$$   \label{lem:ham_expansion}	
\end{lemma}
\begin{proof}
	Consider a set $S$ of size $\bar{\rho} n$, and the first $\alpha n$ vertices in $S$ in a random permutation. This set is just a random subset of $S$ of size $\alpha n$. The number of choices of such subset is $$\binom{\bar{\rho}n}{\alpha n} = 2^{(H(\frac{\alpha}{\bar{\rho}})+o(1))\bar{\rho} n}.$$
	
	Notice that we can apply Lemma~\ref{lem:ham_alpha_beta} for with set $S$, even though $S$ is just a subset of $V$, because the same proof applies only with respect to a subset of vertices in one side of a Hamiltonian graph. Therefore, the number of subsets of size $\alpha n$ of $S$ with at most $\beta n$ neighbors is at most
	\begin{eqnarray*}
		2^{(\alpha H(\frac{\delta}{ \alpha})+(1-\alpha)H(\frac{\delta} {(1-\alpha)})+o(1))\bar{\rho}n}.
	\end{eqnarray*}
	Combining the above, we get that the probability that a random set of $\alpha n$ vertices of $S$ have at most $\beta n$ neighbors is at most
	\begin{eqnarray*}
	2^{(-H(\frac{\alpha}{\bar{\rho}})+\alpha H(\frac{\delta}{ \alpha})+(1-\alpha)H(\frac{\delta} {(1-\alpha)})+o(1))\bar{\rho}n}.
	\end{eqnarray*}
\end{proof}

Now that we have established the three lemmas we are ready to prove Theorem~\ref{thm:random-hamiltonian}.
	
\begin{proof}[Proof of Theorem~\ref{thm:random-hamiltonian}]
	Setting $\epsilon=0.0012$, $\alpha=0.245$ and $\beta=0.3675$ (and $\rho=\frac{1}{2}+\epsilon$, $\bar{\rho}=1-\rho$), we get that these parameters satisfy the conditions for Lemmas \ref{lem:order}, \ref{lem:ham_expansion} and \ref{lem:ham_bad_sets}.
	
	Applying Lemma \ref{lem:ham_bad_sets}, we get that the size of $B_\epsilon$ is at most $n_B\leq 2^{0.044n}$. Applying Lemma~\ref{lem:ham_expansion}, we get that the probability that the lowest $\alpha n$ vertices of a set of size $\bar{\rho}n$ have less than $\beta n$ neighbors is at most $p\leq 2^{-0.86n}$. Applying Lemma~\ref{lem:order}, we get that the probability that for a set $S$ of size $\bar{\rho}n$ the $\alpha n$th vertex in a random $\pi$ comes after $\beta n$ vertices of $V-S$ is at most $q\leq 2^{-0.45n}$. Combining these three, we get that the probability there exists a set of size $\bar{\rho}n$ unmatched by a random $\pi$ is at most $n_B(p+q)<1$, therefore, there must be a $\pi$ that matches at least one vertex in each set of size  $\bar{\rho}n$, and the proof follows.
\end{proof}

This proof approach can be also used to show that a random permutation guarantees to match more than a half of the vertices in every regular graph.
On the other hand, Theorem~\ref{thm:3-4} in Section~\ref{sec:additional} shows that one cannot hope to get $\rho > 3/4$ with a random permutation in regular graphs.

\begin{figure}[h!]
\begin{center}
	\includegraphics[scale=.4]{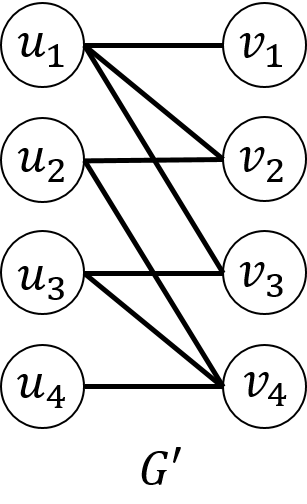}
	\caption{The graph $G'$ admits two bad sets, namely $\{v_1,v_2\}$ and $\{v_1,v_3\}$.}
\end{center}
\label{fig:many-bad}
\end{figure}

The following proposition gives an example of a bipartite graph that admits many bad sets.

\begin{proposition}
There exists a bipartite graph $G(U,V;E)$ with $n$ vertices on each side, which admits $2^{n/4}$ bad sets.
\label{ex:many-bad}
\end{proposition}

\begin{proof}
Consider the bipartite graph $G'$ (with four vertices on each side) depicted in Figure~\ref{fig:many-bad}.
The sets $V_1=\{v_1,v_3\}$ and $V_2=\{v_1,v_2\}$ are bad sets.
Indeed, under $\pi=\{v_4,v_3,v_2,v_1\}$, the set $\{v_1,v_2\}$ remains unmatched if $u_1$ and $u_2$ arrive before $u_3$ and $u_4$.
Similarly, under $\pi=\{v_4,v_2,v_3,v_1\}$, the set $\{v_1,v_3\}$ remains unmatched if $u_1$ and $u_2$ arrive before $u_3$ and $u_4$.
Now, consider the graph $G$ that is composed of $n/4$ disjoint copies of $G'$.
Since each copy admits two bad sets, $G$ admits $2^{n/4}$ bad sets.
\end{proof}

%% file: iterative.tex
\section{Iterative Process}
\label{sec:iterative}

A natural approach for establishing the existence of a good permutation $\pi$ for the max min greedy matching problem is the following iterative process of ``upgrading" unmatched vertices.

Given a permutation $\pi:V\rightarrow [n]$ and a permutation $\sigma: U\rightarrow [n]$, let $M[\pi,\sigma]$ be the result of the greedy matching where vertices in $U$ arrive in order $\sigma$ (from low to high) and each vertex $u\in U$ is matched to its lowest (under $\pi$) neighbor (or left unmatched if all its neighbors are already matched).

\begin{figure}[h!]
\begin{center}
	\includegraphics[scale=.4]{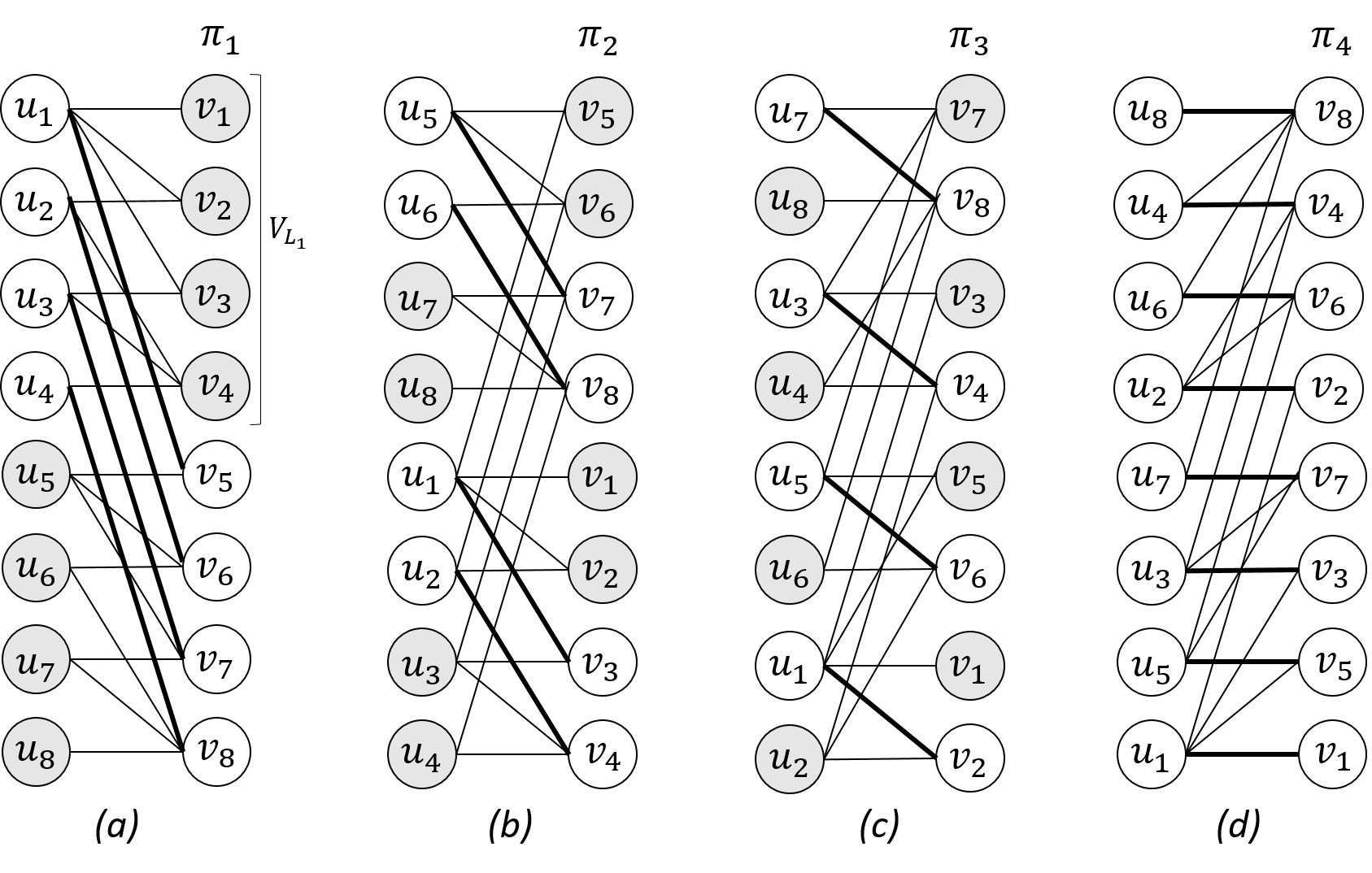}
	\caption{An iterative process where unmatched vertices are given priority. In every iteration thick edges are in the matching; gray vertices are unmatched.}
\end{center}
\label{fig:iterative-log}
\end{figure} 

Fix an arbitrary permutation $\pi_1$ on $V$, and let $\sigma_1$ be a permutation on $U$ minimizing the greedy matching\footnote{It is unclear whether $\sigma_1$ can be computed in polynomial time. The related problem of computing a minimum maximal matching in bipartite graphs is known to be NP-hard \cite{Demange2008}. However, here we consider the existential problem.}.
Let $M_1 = M[\pi,\sigma]$ be the result of the greedy matching under permutations $\sigma$ and $\pi$.
If $|M_1|/n$ is some constant greater than $1/2$, then terminate with permutation $\pi_1$.
Otherwise, partition $V$ into the set $V_{L_1}$ of unmatched vertices ($L$ for {\em low}, as they will be placed low in the next iteration, and also for {\em losers}, or {\em leftovers}) and the set $V_{H_1}$ of matched vertices ($H$ for {\em high}, as they will be placed high in the next iteration, and also for {\em hitters}, or {\em happy}).

Consider now a permutation $\pi_2$ in which $V_{L_1}$ precedes $V_{H_1}$ (preserving the internal order between vertices in $V_{L_1}$ and similarly between vertices in $V_{H_1}$), and let $\sigma_2$ be a permutation on $U$ minimizing the resulting greedy matching.
Let $M_2=M[\pi_2,\sigma_2]$.
If $|M_2|/n$ is some constant greater than $1/2$, then terminate with permutation $\pi_2$.
Else, partition $V$ into the set $V_{L_2}$ of unmatched vertices and the set $V_{H_2}$ of matched vertices, and consider a permutation $\pi_3$ in which $V_{L_2}$ precedes $V_{H_2}$ (preserving internal orders).

Continue this iterative process until the obtained permutation $\pi_k$ ensures a matching greater than a half.

The intuition behind this approach is that the unmatched vertices need some ``help" in order to be matched, and we provide this help in the form of prioritizing them over their mates.
One might hope that this process will reach a good permutation within a constant number of iterations.
Unfortunately, we show an example where the process goes through $\log n$ iterations before it first obtains a permutation ensuring a matching that exceeds $n/2$.

The construction of the graph is inductive.
The base is $G_0(U_0,V_0;E_0)$, with two vertices $u,v$ and a single edge between them.
For every $i=1,2,\ldots$, $G_i(U_i,V_i;E_i)$ is such that $|U_i|=|V_i|=2^i$; it is obtained by taking two (disjoint) copies of $G_{i-1}$, with additional edges of the form $(u_j,v_j)$ for every $u_j$ from one copy of $G_{i-1}$ to $v_j$ in the second copy of $G_{i-1}$.
An example of $G_3$ is presented in Figure \ref{fig:iterative-log}(a).
The iterative process is depicted in Figure \ref{fig:iterative-log}(a)-(d).
In all iterations preceding the last one, exactly $n/2$ vertices are matched in the worst $\sigma$.

%% file: app-additional-examples.tex
\section{Additional Results}
\label{sec:additional}

The following theorem shows that one cannot hope to get $\rho > 3/4$ with a random permutation in regular graphs.

\begin{thm}
\label{thm:3-4}
For every $\epsilon > 0$ and sufficiently large $d$, there are graphs $G$ for which a random permutation $\pi$ results in $\rho \le \frac{3}{4} + \epsilon$.
\end{thm}

\begin{proof}
Consider a $d$-regular bipartite graph $G(U,V;E)$, where $d$ is very large, there is a balanced bipartite independent set $(S,T)$ of size $\frac{1 - \epsilon}{2}n$, and conditioned on that, $G$ is random. Let $Q$ (a random variable) be the set of first $\frac{1 + \epsilon}{2}n$ vertices under the random permutation $\pi$. Then, $E[|T \cap (V \setminus Q)|] =  (\frac{1 - \epsilon}{2})^2 n \simeq \frac {1}{4}n$. W.h.p. there will be a perfect matching between $Q$ and $U \setminus S$. Hence one can choose a permutation $\sigma$ over $U$ that matches all of $U \setminus S$ to $Q$. But then the vertices $T \cap (V \setminus Q)$ will remain unmatched.
\end{proof}

We also establish a few impossibility results for regular graphs of low degree.

\begin{thm}
The following hold:
\begin{itemize}
  \item There exists a $3$-regular bipartite graph $G$ for which $\rho(G) = \frac{5}{7}$.
  \item There exists a $4$-regular bipartite graph $G$ for which $\rho(G) = \frac{10}{13}$.
\end{itemize}
\label{thm:imposs-projective}
\end{thm}

The proof relies on graphs induced by projective planes.
A projective plane consists of a set of lines and a set of points, where (among other properties) every two lines intersect in a single point and every two points are incident to a single line.
A projective plane induces a bipartite graph $G(U,V;E)$, where every vertex $u \in U$ corresponds to a point in the plane, every vertex $v \in V$ corresponds to a line, and there exists an edge between $u$ and $v$ if the point corresponding to $u$ is incident to the line corresponding to $v$.

\begin{proof}
For the first result, we show that $\rho = \frac{5}{7}$ for the bipartite graph induced by the Fano plane.
The Fano plane is a projective plane consisting of 7 points and 7 lines, with 3 points on every line and 3 lines through every point.
Consider the $3$-regular bipartite graph $G(U,V;E)$ induced by the Fano plane.
Let $N(V')$ denote the neighbors of a set $V' \in V$.
For every set $V' \in V$ such that $|V'|=2$, it holds that $|N(V')|=5$.
We show below that for every such $V'$ there exists a perfect matching between $N(V')$ and $V \setminus V'$.
Hence one can choose a permutation $\sigma$ over $U$ whose first $5$ vertices are $N(V')$ that will match the vertices of $V \setminus V'$ one by one. Thereafter, the vertices of $V'$ will remain unmatched.
By Hall's condition, it suffices to show that for every set $U' \subset N(V')$ such that $|U'| \leq 5$ it holds that $|N(U')| \geq |U'|+2$ (so that Hall's condition applies with respect to the set $V \setminus V'$).
Indeed, for every set $U'$ of size 1, $|N(U')|=3$, for every set $U'$ of size $\geq 2$, $|N(U')| \geq 6$, and for every set $U'$ of size $5$, $|N(U')|=7$.
It follows that $\rho(G)=5/7$.

The second result follows a similar argument.
It is known that there exists a projective plane consisting of $13$ points and $13$ lines, with $4$ points on every line and $4$ lines through every point.
We claim that $\rho = \frac{10}{13}$ for the bipartite graph $G(U,V;E)$ induced by this projective plane.
By the properties of a projective plane, for every set $V' \in V$ such that $|V'|=3$, it holds that $|N(V')| \in \{9,10\}$.
We show below that for every such $V'$ there exists a perfect matching between $N(V')$ (and possibly an additional vertex $u$ in case $|N(V')| = 9$) and $V \setminus V'$.
Hence one can choose a permutation $\sigma$ over $U$ whose first $10$ vertices are $N(V')$ (possibly with the additional vertex) that will match the vertices of $V \setminus V'$ one by one. Thereafter, the vertices of $V'$ will remain unmatched.
By Hall's condition, it suffices to show that for every set $U' \subset N(V')$ such that $|U'| \leq 10$ it holds that $|N(U')| \geq |U'|+3$ (so that Hall's condition applies with respect to the set $V \setminus V'$).
Indeed, for every set $U'$ of size 1, $|N(U')|=4$, for every set $U'$ of size $\geq 2$, $|N(U')| \geq 7$, for every set $U'$ of size $\geq 5$, $|N(U')| \geq 11$, and for every set $U'$ of size $\geq 9$, $|N(U')| = 13$,
It follows that $\rho(G)=10/13$.
\end{proof}